\documentclass[%
 aip,
  amsmath,amssymb,
  reprint,%
]{revtex4-1}

\usepackage[margin=0.75in]{geometry}

\usepackage{graphicx}
\usepackage{dcolumn}

\usepackage[english]{babel}
\usepackage[utf8]{inputenc}
\usepackage[T1]{fontenc}
\usepackage{mathptmx}
\usepackage{amssymb}
\usepackage{amsmath}
\usepackage{amsthm}
\usepackage{caption}
\usepackage{mathdots}
\usepackage{bm}

\usepackage{algorithmic}
\usepackage{algorithm}
\usepackage[ruled,vlined,algo2e]{algorithm2e}

\usepackage[all,cmtip]{xy}
\usepackage{tikz}
\usetikzlibrary{automata,arrows, positioning, shapes.geometric, shadows.blur, fit}

\usepackage{lipsum}
\usepackage{mathrsfs}

 \newtheorem{thm}{Theorem}[section]
 \newtheorem{cor}[thm]{Corollary}
 \newtheorem{lem}[thm]{Lemma}
 
 \newtheorem{defn}[thm]{Definition}

 \numberwithin{equation}{section}
 \numberwithin{algorithm}{subsection}

\begin{document}

\preprint{AIP/123-QED}

\begin{titlepage}
\begin{center}
\vspace{3cm}

{\large \textsc{Comisi\'on Nacional de Bancos y Seguros}}\\[1.5cm]
\hrule
\vspace{.5cm}
{\textsc{Working paper:}\\
\Large \bfseries Dynamic financial processes identification using sparse regressive reservoir computers} 
\vspace{.5cm}

\hrule
\vspace{1.5cm}

\includegraphics[width=0.35\textwidth]{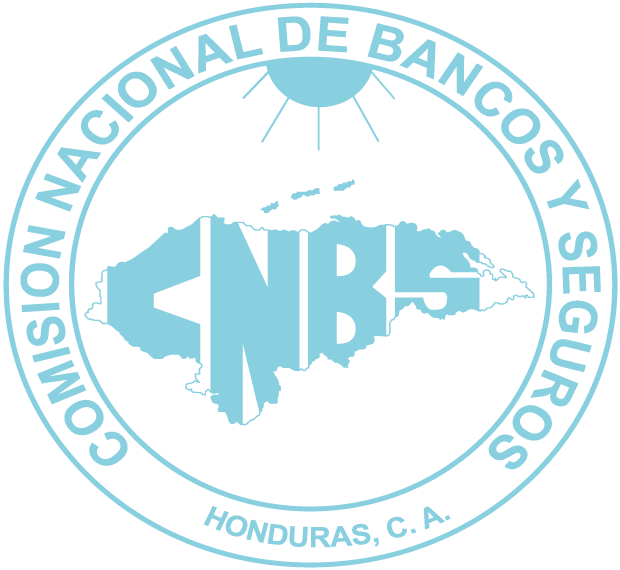}~\\[1cm]
\vspace{2cm}
\vspace{1.5cm}

\textsc{\textbf{Authors}}\\
\vspace{.5cm}
\centering

Fredy Vides\\
Idelfonso B. R. Nogueira\\
Gabriela Lopez Gutierrez\\
Lendy Banegas\\
Evelyn Flores

\vspace{3cm}

\centering \today 
\end{center}

\pagebreak

\end{titlepage}

\title{Dynamic financial processes identification using sparse regressive reservoir computers}

\author{Fredy Vides}
\email{fredy.vides@cnbs.gob.hn}
\affiliation{Center for Innovation in Scientific Computing, Universidad Nacional Aut\'onoma de Honduras, Honduras}

\author{Idelfonso B. R. Nogueira}
\email{idelfonso.b.d.r.nogueira@ntnu.no}    
\affiliation{Department of Chemical Engineering, Norwegian University of Science and Technology, Norway}

\author{Gabriela Lopez Gutierrez}
\email{gabriela.lopez@cnbs.gob.hn}
\affiliation{Department of Statistics and Research, National Commission of Banks and Insurance Companies of Honduras, Honduras}

\author{Lendy Banegas}
\email{lendy.banegas@cnbs.gob.hn}
\affiliation{Department of Statistics and Research, National Commission of Banks and Insurance Companies of Honduras, Honduras}

\author{Evelyn Flores}
\email{evelyn.flores@cnbs.gob.hn}
\affiliation{Department of Statistics and Research, National Commission of Banks and Insurance Companies of Honduras, Honduras}

\date{\today}

\begin{abstract}
In this document, we present key findings in structured matrix approximation theory, with applications to the regressive representation of dynamic financial processes. Initially, we explore a comprehensive approach involving generic nonlinear time delay embedding for time series data extracted from a financial or economic system under examination. Subsequently, we employ sparse least-squares and structured matrix approximation methods to discern approximate representations of the output coupling matrices. These representations play a pivotal role in establishing the regressive models corresponding to the recursive structures inherent in a given financial system. The document further introduces prototypical algorithms that leverage the aforementioned techniques. These algorithms are demonstrated through applications in approximate identification and predictive simulation of dynamic financial and economic processes, encompassing scenarios that may or may not exhibit chaotic behavior.
\end{abstract}

\maketitle

\begin{quotation}
The intricate dynamics inherent in financial processes often pose challenges for accurate modeling and prediction. Nonetheless, the synergy of sparse representation techniques with Nonlinear Regressive Reservoir Computers (NRRCs) proves advantageous in modeling financial processes dynamics. Firstly, this approach excels in capturing the intricate nonlinear dynamics of financial data. NRRCs, adept at modeling complex relationships between input and output data, coupled with sparse representation, effectively identify the key dynamic components, ensuring more accurate and precise modeling of underlying dynamics. Secondly, the methodology promotes efficient data utilization. NRRCs, capable of learning from a relatively small dataset, align well with the limited scope and complexity of financial processes data. By pinpointing crucial variables, the approach enhances modeling efficiency, conserving time and resources. Thirdly, the approach exhibits flexibility and adaptability. NRRCs swiftly respond to changing conditions, making them ideal for the dynamic nature of financial processes. The amalgamation of NRRCs with sparse representation facilitates the identification of changes in the underlying structure, enabling prompt adjustments to the model. In conclusion, integrating sparse representation techniques with time series models employing nonlinear regressive reservoir computers yields several advantages for financial processes dynamics modeling. It ensures accurate modeling of complex dynamics, optimizes data utilization, and provides adaptability to evolving conditions.
\end{quotation}

\section{Introduction}

Regressive models and reservoir computers are robust computational tools for the identification and simulation of financial and economic systems \cite{RC_NLFinancialSystems}. In recent years, a new class of architectures, termed next-generation reservoir computers, has emerged \cite{Gauthier2021}. In this study, we delve into the intrinsic network architecture associated with these reservoir computers, which significantly contribute to data dimensionality reduction. This architecture also facilitates the parametric identification processes by leveraging the matrix structural constraints induced by the network architecture. The document outlines key aspects of the theory and algorithms pertaining to the computation of specific types of regressive reservoir computers. The focus of this study is on reservoir computers, the architecture of which can be approximated by either linear or nonlinear regressive vector models.

The main contribution of the work reported in this document is the application of {\em collaborative schemes} involving structured matrix approximation methods, together with linear and nonlinear regressive models, to the simulation of dynamic financial processes. Some theoretical aspects of the aforementioned methods are described in \S\ref{sec:structured-approximation}. As a byproduct of the work reported in this document, a toolset of Python programs for financial and economic dynamic models identification based on the ideas presented in \S\ref{sec:structured-approximation} and \S\ref{sec:algorithms} has been developed and is available in \cite{FVides_DyNet}.

Even though, the applications of the structure preserving function approximation technology developed as part of the work reported in this document can range from numerical modeling of cyber-physical systems \cite{Yuan2019}, to climate simulation \cite{RC4ClimateSimulation}. We will focus on applications to financial processes identification in this paper.

Financial processes have become complex systems where several dynamic entities constantly communicate and affect each other. Hence, the financial  processes identification has become a critical aspect of modern finance. The identified models can become a helpful tool for institutions to analyze and predict financial trends, manage risk, and make informed investment decisions (Bodie et al., 2014). However, the complexity and uncertainty of financial markets make these tasks challenging. Financial processes often exhibit nonlinear and complex behavior, which makes it difficult to model and identify the underlying dynamics (Cont, 2001). Traditional linear models may fail to capture the intricate relationships between variables, leading to inaccurate predictions and suboptimal decision-making. 

Despite the challenges posed by the factors described above, data quality, and market efficiency, machine learning techniques offer promising solutions for improving the accuracy and utility of financial models. Machine learning has been used to identify the relationship between the key financial ratios that characterize a firm’s financial position. For instance, Dixon, Klabjan, and Bang's \cite{dixon2017classificationbased} work applies deep learning to predict financial market movements. The authors use a classification approach to predict financial market movements. Their findings suggest that deep learning algorithms can provide valuable insights and predictions about financial market movements, outperforming traditional methods. Sirignano and Cont \cite{sirignano2018universal} propose a deep learning model to identify the dynamics of price formation of a high-frequency limit order book. Their model was able to capture universal features of price formation across different markets, highlighting the potential of machine learning to model complex financial systems.

Overall, the recent literature suggests that machine learning has significant potential in modeling financial data. These techniques are increasingly utilized to capture complex patterns, make accurate predictions, and optimize decision-making in the financial domain. However, it is still an open issue to be investigated. In this scenario, this work also contributes to the field of financial data identification by applying the proposed tools in this context leading to a better understanding of the underlying financial processes addressed here.

A prototypical algorithm for the computation of sparse structured recursive models based on the ideas presented in \S\ref{sec:structured-approximation}, is presented in \S\ref{sec:algorithms}. Some numerical simulations of financial processes based on the prototypical algorithm presented in \S\ref{sec:algorithms} are documented in \S\ref{sec:experiments}.

\section{Preliminaries and Notation}
The symbols $\mathbb{R}^+$ and $\mathbb{Z}^+$ will be used to denote the positive real numbers and positive integers, respectively. For any pair $p,n\in \mathbb{Z}^+$ the expression $d_p(n)$ will denote the positive integer $d_p(n)=n(n^{p}-1)/(n-1)+1$. Given $\delta>0$, let us consider the function defined by the expression 
\begin{align*}
H_\delta(x)=\left\{
\begin{array}{ll}
1, & x>\delta\\
0,& x\leq \delta
\end{array}
\right..
\end{align*}
Given a matrix $A\in \mathbb{C}^{m\times n}$ with singular values \cite[\S2.5.3]{GoVa96} denoted by the expressions $s_j(A)$ for $j=1,\ldots,\min\{m,n\}$. We will write $\mathrm{rk}_\delta(A)$ to denote the number
\begin{align*}
\mathrm{rk}_\delta(A)=\sum_{j=1}^{\min\{m,n\}}H_\delta(s_j(A)).
\end{align*}
For a nonzero matrix $A\in \mathbb{R}^{m\times n}$, the symbol $A^+$ will be used to denote the pseudoinverse\cite[\S5.5.4]{GoVa96} of $A$.

Given a scalar time series $\Sigma=\{x_{t}\}_{t\geq 1}\subset \mathbb{R^n}$, a positive integer $L$ and any $t\geq L$, we will write $\mathbf{x}_L(t)$ to denote the vector
\[
\mathbf{x}_L(t)=\begin{bmatrix}
\mathbf{x}^{(1)}_{L}(t)^\top & \mathbf{x}^{(2)}_{L}(t)^\top & \cdots &  \mathbf{x}^{(n)}_L(t)^\top
\end{bmatrix}^\top \in \mathbb{R}^{nL},
\]
with 
\[
\mathbf{x}^{(j)}_{L}(t)=\begin{bmatrix}
x^{(j)}_{t-L+1} & x^{(j)}_{t-L+2} & \cdots & x^{(j)}_{t-1}  &  x^{(j)}_{t}
\end{bmatrix}^\top \in \mathbb{R}^{L}.
\]
for $1\leq j\leq n$, where $x_{j,s}$ denotes the scalar $j$-component of each element $x_s$ in the vector time series $\Sigma$, for $s\geq 1$.

The identity matrix in $\mathbb{R}^{n\times n}$ will be denoted by $I_n$, and we will write $\hat{e}_{j,n}$ to denote the matrices in $\mathbb{R}^{n\times 1}$ representing the canonical basis of $\mathbb{R}^{n}$ (each $\hat{e}_{j,n}$ corresponds to the $j$-column of $I_n$). For any vector $x\in \mathbb{R}^n$, we will write $\|x\|$ to denote the Euclidean norm of $x$. Given a matrix $X \in \mathbb{R}^{m\times n}$, the expression $\|X\|_F$ will denote the Frobenius norm of $X$.

For any integer $n>0$, in this article, we will identify the vectors in $\mathbb{R}^n$ with column matrices in $\mathbb{R}^{n\times 1}$.

Given two matrices $A\in \mathbb{R}^{m\times n}$, $B\in \mathbb{R}^{p\times q}$, the tensor Kronecker tensor product $A\otimes B \in \mathbb{R}^{mp\times nq}$ is determined by the following operation.
\[
A\otimes B = \begin{bmatrix}
a_{11}B & \cdots & a_{1n}B\\
\vdots & \ddots & \vdots\\
a_{m1}B & \cdots & a_{mn}B
\end{bmatrix}
\]
For any integer $p>0$ and any matrix $X\in \mathbb{R}^{m\times n}$, we will write $X^{\otimes p}$ to denote the operation determined by the following expression.
\[
X^{\otimes p} =\left\{
\begin{array}{ll}
X&, p=1\\
X\otimes X^{\otimes (p-1)}&, p\geq 2
\end{array}
\right.
\]
We will also use the symbol $\Pi_p$ to denote the operator $\Pi_p:\mathbb{R}^n\to \mathbb{R}^{n^p}$ that is determined by the expression  $\Pi_p(x):=x^{\otimes p}$, for each $x\in \mathbb{R}^n$. Given two matrices $X=[x_{i,j}]$, $Y=[y_{i,j}]$ in $\mathbb{R}^m$, we will write $X\odot Y$ to denote the operation corresponding to their Hadamard product $X\odot Y :=\begin{bmatrix}
    x_{i,j}y_{i,j}
\end{bmatrix}\in \mathbb{R}^m$.

For any matrix $A\in \mathbb{R}^{m\times n}$, we will denote by $\mathrm{colsp}(A)$ the columns space of the matrix $A$. Given a list $A_1,A_2,\ldots,A_m$ such that for $1\leq j\leq m$, $A_j\in \mathbb{R}^{n_j\times n_j}$ for some integer $n_j>0$. The expression $A_1\oplus A_2 \oplus \cdots \oplus A_m$ will denote the block diagonal matrix
\[
A_1\oplus A_2 \oplus \cdots \oplus A_m=\begin{bmatrix}
A_1 & & & \\
& A_2 & &\\
& & \ddots & \\
& & & A_m
\end{bmatrix},
\]
where the zero matrix blocks have been omitted. 

In this article, we will use the following notion of sparse representation. Given $\delta>0$ and two matrices $A\in \mathbb{R}^{m\times n}$ and $X\in \mathbb{R}^{n\times p}$, a matrix $\hat{X}\in \mathbb{R}^{n\times p}$ is an approximate sparse representation of $X$ with respect to $A$, or a sparse representation of $X$ for short, if $\|\hat{X}A-XA\|_F\leq C\delta$ for some $C>0$ that does not depend on $\delta$, and $\hat{X}$ has fewer nonzero entries than $X$.

We will write $\mathbf{S}^1$ to denote the set $\{z\in \mathbb{C}:|z|=1\}$. Given any matrix $X\in \mathbb{R}^{m\times n}$, we will write $X^\top$ to denote the transpose $X^\top\in \mathbb{R}^{n\times m}$ of $X$. A matrix $P\in \mathbb{C}^{n\times n}$ will be called an orthogonal projector whenever $P^2=P=P^\top$. Given any matrix $A\in \mathbb{R}^{n\times n}$, we will write $\Lambda(A)$ to denote the spectrum of $A$, that is, the set of eigenvalues of $A$.

\section{Structured dynamic transformation model identification}
\label{sec:structured-approximation}

Given two discrete-time dynamic systems determined by two time series $\{x_t\}_{t\geq 1}$ and $\{y_t\}_{t\geq 1}$, respectively. We will study the identification process of maps determined by the expression
\begin{align}
y_t = \mathscr{F}(x_t)+r_t,
    \label{eq:dynamic_transformation_def}
\end{align}
where $\{r_t\}_{t\geq 1}$ denotes the sequence of residual errors determined for each $t\geq 1$ by $r_t:=\|x_t-\mathscr{F}(x_t)\|$ for some suitable norm $\|\cdot\|$.

\subsection{Low-rank approximation and sparse linear least squares solvers}

\label{section:linear-solvers}

In this section, some low-rank approximation methods with applications to the solution of sparse linear least squares problems are presented.

\begin{defn}
Given $\delta>0$ and a matrix $A\in \mathbb{C}^{m\times n}$, we will write $\mathrm{rk}_{\delta}(A)$ to denote the nonnegative integer determined by the expression
\[
\mathrm{rk}_{\delta}(A)=\sum_{j=1}^{\min\{m,n\}} H_\delta(s_j(A)),
\]
where the numbers $s_j(A)$ represent the singular values corresponding to an economy-sized singular value decomposition of the matrix $A$.
\end{defn}

\begin{lem}\label{lem:rk-delta-inv}
We will have that $\mathrm{rk}_\delta\left(A^\top\right)=\mathrm{rk}_\delta(A)$ for each $\delta>0$ and each $A\in \mathbb{C}^{m\times n}$.
\end{lem}
\begin{proof}
Given an economy-sized singular value decomposition 
\begin{align*}
U\begin{bmatrix}
s_1(A) & & &\\
& s_2(A) & &\\
&&\ddots&\\
&&& s_{\min\{m,n\}}(A)
\end{bmatrix}V=A
\end{align*}
we will have that
\begin{align*}
V^\top\begin{bmatrix}
s_1(A) & & &\\
& s_2(A) & &\\
&&\ddots&\\
&&& s_{\min\{m,n\}}(A)
\end{bmatrix}U^\top=A^\top
\end{align*}
is an economy-sized singular value decomposition of $A^\top$. This implies that
\begin{align*}
\mathrm{rk}_{\delta}\left(A^\top\right)=\sum_{j=1}^{\min\{m,n\}} H_\delta(s_j(A))=\mathrm{rk}_{\delta}(A)
\end{align*}
and this completes the proof.
\end{proof}

\begin{lem}\label{lem:rk-lem}
Given $\delta>0$ and $A\in \mathbb{C}^{m\times n}$ we will have that $\mathrm{rk}_\delta(A)\leq \mathrm{rk}(A)$.
\begin{proof}
We will have that $\mathrm{rk}(A)=\sum_{j=1}^{\min\{m,n\}} H_0(s_j(A))\geq \sum_{j=1}^{\min\{m,n\}} H_\delta(s_j(A))=\mathrm{rk}_\delta(A)$. This completes the proof.
\end{proof}
\end{lem}

\begin{thm}\label{thm:main-SpLSSolver}
Given $\delta>0$ and $y,x_1,\ldots,x_m\in \mathbb{C}^n$, let 
\begin{align*}
X&=\begin{bmatrix}
| & | &  & |\\
x_1 & x_2 & \cdots & x_{m}\\
| & | &  & |
\end{bmatrix}.
\end{align*}
If $\mathrm{rk}_\delta\left(X\right)>0$ and if we set $r=\mathrm{rk}_\delta\left(X\right)$ and $s_{n,m}(r)=\sqrt{r(\min\{m,n\}-r)}$ 
then, there are a rank $r$ orthogonal projector $Q$, $r$ vectors $x_{j_1},\ldots,x_{j_r}\in \{x_1,\ldots,x_m\}$ and $r$ scalars $c_1\ldots,c_r\in \mathbb{C}$ such that $\|X-QX\|_F\leq (s_{n,m}(r)/\sqrt{r})\delta$, and $\|y-\sum_{k=1}^r c_kx_{j_k}\|\leq \left(\sum_{k=1}^r|c_k|^2\right)^{\frac{1}{2}}s_{n,m}(r)\delta+\|(I_n-Q)y\|$.
\end{thm}
\begin{proof}
Let us consider an economy-sized singular value decomposition $USV=A$. If $u_j$ denotes the $j$-column of $U$, let $Q$ be the rank $r=\mathrm{rk}_\delta(A)$ orthogonal projector determined by the expression $Q=\sum_{j=1}^r u_ju_j^\ast$. It can be seen that
\begin{align*}
\|X-QX\|_F^2&=\sum_{j=r+1}^{\min\{m,n\}} s_j(X)^2\\
&\leq (\min\{m,n\}-r)\delta^2=\frac{s_{n,m}(r)^2}{r}\delta^2.
\end{align*}
Consequently, $\|X-QX\|_F\leq \frac{s_{n,m}(r)}{\sqrt{r}}\delta$.

Let us set.
\begin{align*}
\hat{X}&=
\begin{bmatrix}
| & | &  & |\\
\hat{x}_1 & \hat{x}_2 & \cdots & \hat{x}_{m}\\
| & | &  & |
\end{bmatrix}=QX\\
\hat{X}_y&=
\begin{bmatrix}
| & | &  & | & |\\
\hat{x}_1 & \hat{x}_2 & \cdots & \hat{x}_{m} & \hat{y}\\
| & | &  & | & |
\end{bmatrix}=Q\begin{bmatrix}
X & y\\
\end{bmatrix}
\end{align*}
Since by lemma \ref{lem:rk-lem} $\mathrm{rk}(X)\geq \mathrm{rk}_\delta(X)$, we will have that $\mathrm{rk}(\hat{X})=r=\mathrm{rk}_\delta(X)>0$, and since we also have that $\hat{x}_1,\ldots,\hat{x}_m,\hat{y}\in \mathrm{span}(\{u_1,\ldots,u_r\})$, there are $r$ linearly independent $\hat{x}_{j_1},\ldots,\hat{x}_{j_r}\in \{\hat{x}_1,\ldots,\hat{x}_m\}$ such that $\mathrm{span}(\{u_1,\ldots,u_r\})=\mathrm{span}(\{\hat{x}_{j_1},\ldots,\hat{x}_{j_r}\})$, this in turn implies that $\hat{y}\in \mathrm{span}(\{\hat{x}_{j_1},\ldots,\hat{x}_{j_r}\})$ and there are $c_1,\ldots,c_r\in \mathbb{C}$ such that $\hat{y}=\sum_{k=1}^r c_k\hat{x}_{j_k}$. It can be seen that for each $z\in \{x_1,\ldots,x_m\}$
\begin{align*}
\|z-Qz\|&\leq\|X-QX\|_F\leq \frac{s_{n,m}(r)}{\sqrt{r}}\delta,
\end{align*} 
and this in turn implies that
\begin{align*}
\left\|y-\sum_{k=1}^r c_kx_{j_k}\right\|&=\left\|y-\sum_{k=1}^r c_k x_{j_k}-\left(\hat{y}-\sum_{k=1}^r c_k \hat{x}_{j_k}\right)\right\|\\
&=\left\|y-\sum_{k=1}^r c_kx_{j_k}-Q\left(y-\sum_{k=1}^r c_k x_{j_k}\right)\right\|\\
&\leq \left(\sum_{k=1}^r|c_k|^2\right)^{\frac{1}{2}}s_{n,m}(r)\delta+\|(I_n-Q)y\|.
\end{align*}
This completes the proof.
\end{proof}

As a direct implication of theorem \ref{thm:main-SpLSSolver} one can obtain the following corollary.

\begin{cor}\label{cor:SpLSSolver}
Given $\delta>0$, $A\in \mathbb{C}^{m\times n}$ and $y\in \mathbb{C}^m$. If $\mathrm{rk}_\delta\left(A\right)>0$ and if we set $r=\mathrm{rk}_\delta\left(A\right)$ and $s_{n,m}(r)=\sqrt{r(\min\{m,n\}-r)}$ then, there are $x\in\mathbb{C}^n$ and a rank $r$ orthogonal projector $Q$ that does not depend on $y$, such that $\|Ax-y\|\leq \|x\|s_{n,m}(r)\delta+\|(I_m-Q)y\|$ and $x$ has at most $r$ nonzero entries.
\end{cor}
\begin{proof}
Let us set $x=\mathbf{0}_{n, 1}$ and $a_j=A\hat{e}_{j,n}$ for $j=1,\ldots,n$. Since $r=\mathrm{rk}_\delta\left(A\right)>0$ and $s_{n,m}(r)=\sqrt{r(\min\{m,n\}-r)}$, by theorem \ref{thm:main-SpLSSolver} we will have that there is a rank $r$ orthogonal projector $Q$ such that $\|A-QA\|_F\leq (s_{n,m}(r)/\sqrt{r})\delta$, and without loss of generality $r$ vectors $a_{j_1},\ldots,a_{j_r}\in \{a_1,\ldots,a_n\}$ and $r$ scalars $c_1\ldots,c_r\in \mathbb{C}$ with $j_1\leq j_2\leq \cdots\leq j_r$ (reordering the indices $j_k$ if necessary), such that $\|y-\sum_{k=1}^r c_ka_{j_k}\|\leq \left(\sum_{k=1}^r|c_k|^2\right)^{\frac{1}{2}}s_{n,m}(r)\delta+\|(I_m-Q)y\|$. If we set $x_{j_k}=c_{k}$ for $k=1,\ldots,r$, we will have that $\|x\|=\left(\sum_{k=1}^r|c_k|^2\right)^{\frac{1}{2}}$ and $Ax=\sum_{k=1}^r x_{j_k}a_{j_k}=\sum_{k=1}^r c_ka_{j_k}$. Consequently, $\|Ax-y\|\leq \|x\|s_{n,m}(r)\delta+\|(I_m-Q)y\|$. This completes the proof.
\end{proof}

The results and ideas presented in this section can be translated into a sparse linear least squares solver algorithm described by algorithm \ref{alg:main_SLMESolver_alg_1} in \S \ref{sec:algorithms}.

\subsection{Sparse structured nonlinear regressive model identification}
\label{sec:structured_NAR_models}

Given time series data sets $\Sigma_x=\{x_t\}_{t\geq 1}$ and $\Sigma_y=\{y_t\}_{t\geq 1}$ in ${R}^n$ corresponding to the orbits of two discrete-time dynamic financial systems of interest, let us consider the problem of identifying a map $\mathscr{T}$ relating the time series data according to the expression
\begin{equation}
y_{t} = \mathcal{T}(x_t) +r_t,
\label{eq:original_system_dynamics_diff_eq}
\end{equation}
where $r_t$ is some suitable small residual term defined as in \eqref{eq:dynamic_transformation_def}. One may need to preprocess the time series data before proceeding with the approximate representation of a suitable evolution operator. For this purpose, given some prescribed suitable integer $L>0$, one can consider the time series $\mathcal{D}_L(\Sigma_x)$ and $\mathcal{D}_L(\Sigma_y)$ determined by the expressions.
\begin{align*}
   \mathcal{D}_L(\Sigma_x)&=\{\mathbf{x}_L(t)\}_{t\geq L}\\ 
   \mathcal{D}_L(\Sigma_y)&=\{\mathbf{y}_L(t)\}_{t\geq L}
\end{align*}
For the dilated time series $\mathcal{D}_L(\Sigma_x)$ and $\mathcal{D}_L(\Sigma_y)$, the identification process corresponding to the relation \eqref{eq:original_system_dynamics_diff_eq}, can be translated into the approximate solution of equations of the form
\begin{align}
\mathbf{y}_L(t) = \tilde{\mathcal{T}}(\mathbf{x}_L(t)),
\label{eq:dilated_dynamics_diff_eq}
\end{align}
for $t\geq L$. Where $\tilde{\mathcal{T}}$ is the mapping to be approximately identified.

For any $p\geq 1$, let us consider the map $\eth_p:\mathbb{R}^n\to \mathbb{R}^{d_p(n)}$ for $d_p(n)=n(n^{p}-1)/(n-1)+1$, that is determined by the expression. 
\[
\eth_p(x):=
\begin{bmatrix}
\Pi_1(x)\\
\Pi_2(x)\\
\vdots\\
\Pi_p(x)\\
1
\end{bmatrix}=
\begin{bmatrix}
x^{\otimes 1}\\
x^{\otimes 2}\\
\vdots\\
x^{\otimes p}\\
1
\end{bmatrix}
\]

Given integers $p,L>0$, and two orbits $\Sigma_x=\{x_t\}_{t\geq 1}$ and $\Sigma_y=\{y_t\}_{t\geq 1}$ in  $\mathbb{R}^n$, corresponding to two related dynamic financial processes of interest. For finite samples $\Sigma^x_N = \{x_t\}_ {t=1}^T\subset \Sigma_x$ and $\Sigma^y_N = \{y_t\}_ {t=1}^T\subset \Sigma_y$, let us  consider the matrices:
\begin{align}
\mathbf{H}^{(0,p)}_{L}(\Sigma^x_T)&=\begin{bmatrix}
\eth_p(\mathbf{x}_L(L)) & \cdots & \eth_p(\mathbf{x}_L(T))
\end{bmatrix}\label{eq:structured_data_matrices}\\
\mathbf{H}^{(1)}_{L}(\Sigma^y_T)&=\begin{bmatrix}
\mathbf{y}_L(L) & \cdots & \mathbf{y}_L(T)
\end{bmatrix}\nonumber
\end{align}

The mapping identification mechanism used in this study for dilated systems of the form \eqref{eq:dilated_dynamics_diff_eq}, will be approximately described by the expression:
\begin{equation}
\mathbf{y}_L(t)=
\hat{\mathcal{T}}(\mathbf{x}_L(t))=W\eth_p(\mathbf{x}_L(t)), \:\: t\geq L,
\label{eq:evolution_op_id}
\end{equation}
for some matrix $W\in \mathbb{R}^{n\times d_p(n)}$ to be determined, with $d_p(n)=n(n^{p}-1)/(n-1)+1$. Applying the techniques and ideas previously presented in this section, the matrix $W$ in \eqref{eq:evolution_op_id} can be estimated by approximately solving the matrix equation
\begin{equation}
W\mathbf{H}^{(0,p)}_{L}(\Sigma^x_T) = \mathbf{H}^{(1)}_{L}(\Sigma^y_T).
\label{eq:evol_matrix_eq}
\end{equation}
The devices described by \eqref{eq:evolution_op_id} are called regressive reservoir computers (RRC) in this paper. 

For any given integers $L,n,p>0$. Taking advantage of the maps $\eth_p$, one can find an integer $0<r_p(n)<d_p(n)$ together with a sparse matrix $R_{p,L}(n)\in\mathbb{R}^{r_p(n)\times d_p(n)}$, such that $R_{p,L}(n)^+R_{p,L}(n)\eth_p(x)\approx \eth_p(x)$ for $x\in \mathbb{R}^{nL}$. The existence of the pair $r_p(n),R_{p,L}(n)$ is determined by the following theorem.

\begin{thm}\label{thm:Compression_existence}
Given positive integers $n,p$. There are an integer $0<\rho_p(n)<d_p(n)$ and a sparse matrix $R_{p}(n)\in \mathbb{R}^{\rho_p(n)\times d_p(n)}$ with $d_p(n)$ nonzero entries, such that $R_{p}(n)\: \tilde{\eth}_p(x)$ has the least number of non-redundant words (monomial terms) for any $x\in \mathbb{R}^{n}$.
\end{thm}
\begin{proof}
Let $n,p$ be positive integers. Consider the structured embedding map $\tilde{\eth}_p:\mathbb{R}^{n} \to \mathbb{R}^{d_p(n)}$, where $d_p(n) = n(n^p - 1)/(n - 1) + 1$ corresponds to the total number of distinct tensor monomials up to degree $p$, plus a constant term.

We aim to construct a sparse matrix $R_{p}(n) \in \mathbb{R}^{\rho_p(n) \times d_p(n)}$ that maps the embedding $\tilde{\eth}_p(x)$ to a stochastic vector of reduced dimension, while preserving all non-redundant monomial terms.

Let us start by defining the matrix $R\in \mathbb{R}^{1\times d_p(n)}$, with $1$ in its $R_{11}$ entry and with all other entries equal to zero.

For each $2\leq j\leq d_p(n)$, let us consider the indices $j=k_1(j)< k_2(j)<\cdots <k_{n_j}(j)<d_p(n)$ that correspond to the same monomial in $\tilde{\eth}_p(x)$, let us define the matrix $R_0\in \mathbb{R}^{1\times d_p(n)}$ with $1$ in its $R_{1k_l(j)}$ entries for $1\leq l\leq n_j$, and with all other entries equal to zero. Let us now define the augmented matrix 
$$
R:=\begin{bmatrix}
    R\\
    R_0
\end{bmatrix}
$$
Finally, update the matrix $R$, using the operation:
$$
R:=\begin{bmatrix}
    R\\
    R'
\end{bmatrix}
$$
where $R'\in \mathbb{R}^{1\times d_p(n)}$ is the matrix with entry $R'_{1d_p(n)}$ equal to $1$, and with all other entries equal to $0$. 

Let us set $R_p(n):=R$. It can be seen by the way $R$ has been constructed, that the operation $R_p(n)\tilde{\eth}_p(x)$ assigns each group of duplicates in $\tilde{\eth}_p(x)$ to a single representative coordinate of $R_p(n)\tilde{\eth}_p(x)$. This selection is performed by adding over the redundant entries and projecting onto a reduced subspace. Let us set $\rho_p(n)$ as the number of rows of $R$. Because of this, it is clear that $R_{p}(n) \tilde{\eth}_p(x)$ has the least number of non-redundant words. This completes the proof.
\end{proof}

In order to reduce to computational effort corresponding to the solution of \eqref{eq:evol_matrix_eq}, using the matrix $R_{p,L}(n)$ described by Theorem \ref{thm:Compression_existence}, one can obtain an approximate reduced representation of \eqref{eq:evol_matrix_eq} determined by the expression.
\begin{equation}
\bar{W} R_{p,L}(n)\mathbf{H}^{(0,p)}_{L}(\Sigma^x_T) = \mathbf{H}^{(1)}_{L}(\Sigma^y_T)
\label{eq:reduced_evol_matrix_eq}
\end{equation}

The architecture of the regressive reservoir computers considered in this study was inspired by next generation reservoir computers \cite{Gauthier2021}.

Schematically, the regressive models considered in this study can be described by a block diagram of the form,
\begin{equation}
\tikzstyle{sensor}=[draw, thick, text width=3em, 
    text centered, rounded corners, minimum height=2.5em, fill=blue!20, blur shadow={shadow blur steps=5}]
\tikzstyle{ann} = [above, text width=5em, text centered]
\tikzstyle{wa} = [sensor, text width=2em, 
    minimum height=14em, rounded corners, fill=green!30]
\tikzstyle{sc} = [sensor, text width=13em, fill=red!20, 
    minimum height=10em, rounded corners, drop shadow]
\def\blockdist{2.3}        
\begin{tikzpicture}
    \node (wa) [wa]  {$\mathbf{W}$};
    \path (wa.west)+(-3.0,1.75) node (asr1) [sensor, fill=orange!30] {$\boldsymbol{\Pi}_1$};
    \path (wa.west)+(-3.0,0.5) node (asr2)[sensor, fill=orange!50] {$\boldsymbol{\Pi}_2$};
    \path (wa.west)+(-3.0,-1.0) node (dots)[ann] {$\vdots$}; 
    \path (wa.west)+(-3.0,-1.75) node (asr3)[sensor, fill=orange!70] {$\boldsymbol{\Pi}_p$};       
    \path (wa.east)+(1.1,0) node (vote) {$\hat{\mathbf{y}}_L(t)$};
\node[draw, dashed, fit=(asr1) (asr2) (asr3), inner sep=8pt] (fit) {};
    
    \path [draw, ->, thick] (asr1.east) -- (wa.west|-asr1.east);
    \path [draw, ->, thick] (asr2.east) -- (wa.west|-asr2.east);
    \path [draw, ->, thick] (asr3.east) -- (wa.west|-asr3.east);
     \path (wa.west)+(-0.5,-1.0) node (dots) [ann] {$\vdots$};
    \path [draw, ->, thick] (wa.east) -- (vote.west);
    \path [draw, ->, thick] (asr1.west)+(-1,0) node [left] {$\mathbf{x}_L(t)$} -- (asr1.west);        
    \draw [->, thick] (asr1.west)+(-0.5,0) |- (asr2.west);
    \draw [->, thick] (asr2.west)+(-0.5,0) |- (asr3.west);
\end{tikzpicture}
\label{eq:Model_diagram}
\end{equation}
where for each $t\geq L$, the block $\mathbf{W}$ is determined by the expression
\begin{align*}
\mathbf{W}(\Pi_1(\mathbf{x}_L(t)),\ldots,\Pi_p(\mathbf{x}_L(t)))&:=\tilde{W}\begin{bmatrix}
\Pi_1(\mathbf{x}_L(t))\\
\vdots\\
\Pi_p(\mathbf{x}_L(t))
\end{bmatrix}
+c_W\\
&= \begin{bmatrix}
\tilde{W} & c_W
\end{bmatrix}\eth_p\left(\mathbf{x}_L(t)\right)
\end{align*}
and where the matrix $W=\begin{bmatrix}
\tilde{W} & c_W
\end{bmatrix}$ is determined by \eqref{eq:evol_matrix_eq}.

The structure of the generic block $\mathbf{W}$ ub \eqref{eq:Model_diagram} can be factored in the form
\begin{equation}
\tikzstyle{block} = [draw, fill=white, rectangle, 
    minimum height=3em, minimum width=6em]
\tikzstyle{sum} = [draw, fill=white, circle, node distance=1cm]
\tikzstyle{input} = [coordinate]
\tikzstyle{output} = [coordinate]
\tikzstyle{pinstyle} = [pin edge={to-,thin,black}]
\begin{tikzpicture}[auto, node distance=2.5cm,>=latex']

    \node [input, name=input] {};
    \node (u0) [coordinate, xshift = -3.7cm , yshift = 0cm] {};
    \node (u1) [coordinate, yshift = -0.3cm] {};
    
    \node [block, right of=input, node distance=2cm,rounded corners, fill=green!70,blur shadow={shadow blur steps=5}] (K) {$\hat{\mathbf{W}}$};
    \node [block, left of=input, node distance=1.2cm,rounded corners,fill=gray!60,blur shadow={shadow blur steps=5}] (L) {$\mathbf{R}$};

    \draw [->] (input) -- node[name=u] {$\mathbf{y}(t)$} (K);   
    \draw [->] (u0) -- node[name=u] {$\eth_p(\mathbf{x}_L(t))$} (L);

    \node [output, right of=K] (output) {};
    \draw [->] (K) -- node [name=y] {$\hat{\mathbf{y}}_L(t)$}(output); 
    \node [near end] {$ $} (input); 
\end{tikzpicture}
\label{eq:Adapted_coupling_block}
\end{equation}
The layers $\mathbf{R}$ and $\hat{\mathbf{W}}$ of the device \eqref{eq:Adapted_coupling_block} are determined by the expressions
\begin{align*}
\mathbf{R}(\mathbf{x})&=\hat{R}\mathbf{x},\\
\hat{\mathbf{W}}(\mathbf{y})&=W\mathbf{y}
\end{align*}
for any pair of suitable vectors $\mathbf{x},\mathbf{y}$. Where $W$ is a sparse representation of an approximate solution to \eqref{eq:reduced_evol_matrix_eq} and $\hat{R}$ is determined by Theorem \ref{thm:Compression_existence}.


Using the reservoir computer models described by \eqref{eq:evolution_op_id}, \eqref{eq:Model_diagram} and \eqref{eq:Adapted_coupling_block}, we can compute approximate representations of the mappings that satisfy \eqref{eq:dilated_dynamics_diff_eq} using the expression
\begin{align}
\hat{\mathcal{T}}(\mathbf{x}_L(t))&:=\hat{K} \left(\hat{\mathbf{W}}\circ \mathbf{R}\circ \eth_p (\mathbf{x}_L(t)\right)\nonumber\\ &=\hat{K}W\hat{R}\eth_p\left(\mathbf{x}_L(t))\right),
\label{eq:AR_part_Matrix_form}
\end{align}
for each $t\geq L$, with
\[
\hat{K}=\begin{bmatrix}
\hat{e}_{1,nL}^\top\\
\hat{e}_{L+1,nL}^\top\\
\vdots\\
\hat{e}_{(n-1)L+1,nL}^\top
\end{bmatrix}.
\] 
Furthermore, we can use the identified RRC model $\hat{\mathcal{T}}$ to simulate the behavior $y_{t}=\mathcal{T}(x_t)$ of the system described by \eqref{eq:original_system_dynamics_diff_eq} for $L\leq t\leq \tau$, by performing the operation:
\begin{equation}
\mathbf{T}\left(\mathbf{x}_L(t)\right):= \hat{K}\hat{\mathcal{T}}(\mathbf{x}_L(t))=\hat{K}W\eth_p\left(\mathbf{x}_L(t))\right),
\label{eq:AR_part_Matrix_form-2}
\end{equation}
for some suitable $\tau>0$.

\begin{thm}\label{thm:thm-1}
Given $\delta>0$, two integers $p,L>0$, a sample $\Sigma_{T}=\{x_t\}_{t=1}^T$ from a dynamic financial system's orbit $\Sigma=\{x_t\}_{t\geq 1}\subset \mathbb{R}^{n}$ with $T>L$, and a matrix solvent $\bar{W}\in \mathbb{R}^{nL\times r_p(nL)}$ of \eqref{eq:reduced_evol_matrix_eq} with $R_{p,L}(n)$ and $r_p(n)$ determined by Theorem \ref{thm:Compression_existence}. If $r=\mathrm{rk}_\delta(R_{p,L}(n)\mathbf{H}^{(0,p)}_{L}(\Sigma_T))>0$, then there is a sparse representation $\hat{W}\in \mathbb{R}^{nL\times \rho_p(nL)}$ of $\bar{W}$ with at most $r \rho_p(nL)$ nonzero entries such that 
\begin{equation}
\|\hat{W}R_{p,L}(n)\mathbf{H}^{(0,p)}_{L}(\Sigma_T) - \bar{W}R_{p,L}(n)\mathbf{H}^{(0,p)}_{L}(\Sigma_T)\|_F\leq K\delta,
\label{eq:thm1-first-estimate}
\end{equation}
for $K = \sqrt{nL(\min\{\rho_p(nL),T-L\}-r)}(\sqrt{r}\|\hat{W}\|_F+\|\bar{W}\|_F)$, where $\rho_p(nL)$ is the integer described by Theorem \ref{thm:Compression_existence}.
\end{thm}
\begin{proof}
Let us set $H = R_{p,L}(n)\mathbf{H}^{(0,p)}_{L,G}(\Sigma_T)^\top$ and $Y=H\bar{W}^\top$. It suffices to prove that there is a sparse representation $\hat{W}\in \mathbb{R}^{nL\times \rho_p(nL)}$ with at most $r \rho_p(nL)$ nonzero entries such that 
$$\|H\hat{W}^\top - Y\|_F\leq K\delta.$$
Since we have that
\begin{align*}
\mathrm{rk}_\delta(H)&=\mathrm{rk}_\delta\left(\left(R_{p,L}(n)\mathbf{H}^{(0,p)}_{L,G}(\Sigma_T)\right)^\top\right)\\
&=\mathrm{rk}_\delta(R_{p,L}(n)\mathbf{H}^{(0,p)}_{L,G}(\Sigma_T))>0    
\end{align*}
by Lemma \ref{lem:rk-delta-inv}. By Corollary \ref{cor:SpLSSolver}, if we set $r=\mathrm{rk}_\delta\left(H\right)$ and $\alpha=\sqrt{r(\min\{\rho_p(nL),T-L\}-r)}$. We will have that there is a rank $r$ orthogonal projector $Q$ such that for each $j=1,\ldots,nL$, there is $\hat{v}_j\in\mathbb{R}^{nL}$ with at most $r$ nonzero entries, for which $\|H\hat{v}_j-Y\hat{e}_{j,M}\|\leq \alpha\|\hat{v}_j\|\delta+\|(I_{T-L}-Q)Y\hat{e}_{j,nL}\|$. Consequently, if we set
\begin{align*}
\hat{W}=\begin{bmatrix}
| & & |\\
\hat{v}_1 & \cdots & \hat{v}_{nL}\\
| & & |
\end{bmatrix}^\top
\end{align*}
we will have that $\hat{W}$ has at most $nrL$ nonzero entries and
\begin{align*}
\|H\hat{W}^\top-Y\|_F^2&= \sum_{j=1}^{nL} \|H\hat{v}_j-Y\hat{e}_{j,nL}\|^2\\
&\leq M(\alpha\|\hat{W}\|_F \delta+\|(I_{T-L}-Q)Y\|_F)^2,
\end{align*}
and this in turn implies that,
\begin{align}
\|H\hat{W}^\top-Y\|_F&\leq \sqrt{nL}(\alpha\|\hat{W}\|_F \delta+\|(I_{T-L}-Q)H\|_F\|\bar{W}\|_F).
\label{eq:main_ineq}
\end{align}
By \eqref{eq:main_ineq} and by Theorem \ref{thm:main-SpLSSolver} we will have that
\begin{align*}
\|H\hat{W}^\top-Y\|_F&\leq \sqrt{nL}(\alpha\|\hat{W}\|_F \delta+(\alpha/\sqrt{r})\|\bar{W}\|_F\delta)\\
&=\alpha\sqrt{(nL/r)}(\sqrt{r}\|\hat{A}\|_F+\|A\|_F)\delta=K\delta.
\end{align*}
This completes the proof.
\end{proof}

\subsubsection{Sparse structured nonlinear autoregressive model identification}

Given some time series data $\Sigma\subset \mathbb{R}^n$ corresponding to an orbit determined by the difference equation 
\begin{equation}
x_{t+1} = \mathcal{A}(x_t),
\label{eq:original_system_dynamics_diff_eq}
\end{equation}
for some discrete-time dynamic financial model $(\hat{\boldsymbol{\Sigma}},\mathcal{T})$ to be identified. One can use the methods presented in \S\ref{sec:structured_NAR_models} to identify the mapping $\mathscr{S}$, by considering the RRC model identification determined by the problem
\begin{align*}
    y_t = \mathscr{A}(x_t),
\end{align*}
for the time series $\Sigma_x:=\{x_t\}_{t\geq 1}$ and $\Sigma_y:=\{y_t\}_{t\geq 1}$ in $\mathbb{R}^{n}$, with $y_t:=x_{t+1}$ for each $t\geq 1$.

\section{Algorithms}
\label{sec:algorithms}

The sparse model identification methods presented in \S\ref{section:linear-solvers} can be translated into prototypical algorithms that will be presented in this section, some programs for data reading and writing, synthetic signals generation, and predictive simulation are also included as part of the {\bf DyNet} tool-set available in \cite{FVides_DyNet}.

\subsection{Sparse linear least squares solver and structured assembling matrix identification algorithms}

As an application of the results and ideas presented in \S\ref{section:linear-solvers} one can obtain a prototypical sparse linear least squares solver algorithm like algorithm \ref{alg:main_SLMESolver_alg_1}.

\begin{figure}[!tph]
\begin{algorithm}[H]
\begin{flushleft}
\caption{{\bf SLRSolver}: Sparse linear least squares solver algorithm}
\label{alg:main_SLMESolver_alg_1}
\end{flushleft}
\begin{algorithmic}
\STATE{{\bf Data:}\:\:\: $A\in \mathbb{C}^{m\times n}$, $Y\in \mathbb{C}^{m\times p}$, $\delta>0$, $N\in \mathbb{Z}^+$, $\varepsilon>0$}
\STATE{{\bf Result:}\:\:\: $X=\mathbf{SLRSolver}(A,Y,\delta,N,\varepsilon)$}
\begin{enumerate}
\STATE{Compute economy-sized SVD $USV=A$\;}
\STATE{Set $s=\min\{m,n\}$\;}
\STATE{Set $r=\mathrm{rk}_\delta(A)$\;}
\STATE{Set $U_\delta=\sum_{j=1}^r U\hat{e}_{j,s}\hat{e}_{j,s}^\ast$\;}
\STATE{Set $T_\delta=\sum_{j=1}^r (\hat{e}_{j,s}^\ast S\hat{e}_{j,s})^{-1} \hat{e}_{j,s} \hat{e}_{j,s}^\ast$\;}
\STATE{Set $V_\delta=\sum_{j=1}^r \hat{e}_{j,s}\hat{e}_{j,s}^\ast V$\;}
\STATE{Set $\hat{A}=U_\delta^\ast A$\;}
\STATE{Set $\hat{Y}=U_\delta^\ast Y$\;}
\STATE{Set $X_0=V_\delta^\ast T_\delta \hat{Y}$\;}
\FOR{$j=1,\ldots,p$}
\STATE{Set $K=1$\;} 
\STATE{Set $\mathrm{error}=1+\delta$\;}
\STATE{Set $c=X_0\hat{e}_{j,p}$}
\STATE{Set $x_0=c$}
\STATE{Set $\hat{c}=\begin{bmatrix}
\hat{c}_1 & \cdots & \hat{c}_n
\end{bmatrix}^\top=\begin{bmatrix}
|\hat{e}_{1,n}^\ast c| & \cdots & |\hat{e}_{n,n}^\ast c|
\end{bmatrix}^\top$\;}
\STATE{Compute permutation $\sigma:\{1,\ldots,n\}\to \{1,\ldots,n\}$ such that: $\hat{c}_{\sigma(1)}\geq \hat{c}_{\sigma(2)}\geq \cdots \geq \hat{c}_{\sigma(n)}$}
\STATE{Set $N_0=\max\left\{\sum_{j=1}^n H_\varepsilon\left(\hat{c}_{\sigma(j)}\right),1\right\}$\;}
\WHILE{$K\leq N$ \AND $\mathrm{error}>\delta$} 
\STATE{Set $x=\mathbf{0}_{n,1}$\;}
\STATE{Set $A_0=\sum_{j=1}^{N_0} \hat{A}\hat{e}_{\sigma(j),n}\hat{e}_{j,N_0}^\ast$}
\STATE{Solve $c=\arg\min_{\tilde{c}\in \mathbb{C}^{N_0}}\|A_0\tilde{c}-\hat{Y}\hat{e}_{j,p}\|$\;}
\STATE{
\FOR{$k=1,\ldots,N_0$} 
\STATE{Set $x_{\sigma(k)}=\hat{e}_{k,N_0}^\ast c$\;}
\ENDFOR \;}
\STATE{Set $\mathrm{error}=\|x-x_0\|_\infty$\;}
\STATE{Set $x_0=x$\;}
\STATE{Set $\hat{c}=\begin{bmatrix}
\hat{c}_1 & \cdots & \hat{c}_n
\end{bmatrix}^\top=\begin{bmatrix}
|\hat{e}_{1,n}^\ast x| & \cdots & |\hat{e}_{n,n}^\ast x|
\end{bmatrix}^\top$\;}
\STATE{Compute permutation $\sigma:\{1,\ldots,n\}\to \{1,\ldots,n\}$ such that: $\hat{c}_{\sigma(1)}\geq \hat{c}_{\sigma(2)}\geq \cdots \geq \hat{c}_{\sigma(n)}$}
\STATE{Set $N_0=\max\left\{\sum_{j=1}^n H_\varepsilon\left(\hat{c}_{\sigma(j)}\right),1\right\}$}
\STATE{Set $K=K+1$\;} 
\ENDWHILE
\STATE{Set $x_j=x$\;}
\ENDFOR
\STATE{Set $X=\begin{bmatrix}
| & | &  & |\\
x_1 & x_2 & \cdots & x_p\\
| & | &  & |
\end{bmatrix}$\;}
\end{enumerate}
\RETURN $X$
\end{algorithmic}
\end{algorithm}
\end{figure}

The least squares problems $c=\arg\min_{\hat{c}\in \mathbb{C}^K}\|\hat{A}\hat{c}-y\|$ to be solved as part of the process corresponding to algorithm \ref{alg:main_SLMESolver_alg_1} can be solved with any efficient least squares solver available in the language or program where the sparse linear least squares solver algorithm is implemented. For the Python version of algorithm \ref{alg:main_SLMESolver_alg_1} the function {\tt lstsq} is implemented.

In this section, we focus on the applications of the structured matrix approximation methods presented in S\ref{sec:structured-approximation}, to  dynamical financial systems identification via regressive reservoir computers.

\begin{figure}[!ht]
\begin{algorithm}[H]
\begin{flushleft}

\caption{Compression matrix computation algorithm}
\label{alg:comp_matrix}
\end{flushleft}
\begin{algorithmic}
\STATE{{\bf Data:}\:\:\: $n,p,L\in \mathbb{Z}^+$, $\nu,\varepsilon\in \mathbb{R}^+$.\;}
\STATE{{\bf Result:}\:\:\: {\sc Compression matrix factor: } $R_{p,L}(n)$\;}
\begin{enumerate}
\STATE{Choose $nL$ pseudorandom numbers $\hat{x}_1,\ldots,\hat{x}_{nL}\in \mathbb{R}$ from $N(0,1)$\;}
\STATE{Set $\mathbf{y}=\nu \begin{bmatrix}
\hat{x}_1 & \hat{x}_2 & \cdots & \hat{x}_{nL}
\end{bmatrix}^\top$\;}
\STATE{Set $d= d_p(n)$\;}
\STATE{Set $\tilde{\mathbf{x}}=\begin{bmatrix}
\tilde{x}_1 & \cdots & \tilde{x}_d
\end{bmatrix}^\top := \eth_p\left(\mathbf{y}\right)$\;}
\STATE{Choose a pseudorandom number $\alpha \in N(0,1)$;}
\STATE{Set $\tilde{x}_d:=\alpha$\;}
\STATE{Set $R = e_{1,d}^\top$\;}
\FOR{$j=2,\ldots,d$}
\STATE{Find $1\leq k_1,\ldots,k_{n_j}\leq d$ such that $|\tilde{x}_j-\tilde{x}_{k_{m}}|\leq \varepsilon$, for each $1\leq m \leq n_j$\;}
\IF{$k_1 = j$}

\STATE{Set $R_0 := (1/n_j)\sum_{l=1}^{n_j} \hat{e}_{k_l,d}^T$\;}
\STATE{Set $R := \begin{bmatrix}
R \\
R_0
\end{bmatrix}$\;}

\ENDIF
\ENDFOR

\STATE{Set $R_{p,L}(n):=R$\;}
\end{enumerate}
\RETURN $R_{p,L}(n)$
\end{algorithmic}
\end{algorithm}
\end{figure}

\subsection{Structured coupling matrix identification algorithm}

Given a discrete-time dynamic financial model $(\Sigma,\mathscr{T})$ and a structured data sample $\Sigma_T\subset \Sigma$, we can apply Algorithm \ref{alg:comp_matrix} and Algorithm \ref{alg:main_SLMESolver_alg_1}, in order to compute the output coupling matrix that can be used to obtain an approximate representation of the evolution operator $\mathcal{T}$, corresponding to the orbit $\Sigma$. For this purpose, one can use the following Algorithm.

\begin{figure}[!tph]
\begin{algorithm}[H]
\begin{flushleft}
\caption{{\bf RRC Model}: RRC model identification}
\label{alg:main_ERRC_alg_1}
\end{flushleft}
\begin{algorithmic}
\STATE{{\bf Data:}\:\:\:$\Sigma^x_{N}=\{x_{t}\}_{t=1}^{T},\Sigma^y_{N}=\{y_{t}\}_{t=1}^{T}\subset \mathbb{R}^n$\;}
\STATE{{\bf Result:}\:\:\: {\sc Output coupling and compression matrices: $\hat{W},\tilde{W},R_{p,L}(n)$\;}}
\begin{enumerate}
\STATE{Choose or estimate the lag value $L$ using auto-correlation function based methods\;}
\STATE{Set a tensor order value $p$\;}
\STATE{Compute compression matrix $R_{p,L}(n)$ applying Algorithm \ref{alg:comp_matrix}\;}
\STATE{Compute matrices:
\begin{align*}
\mathbf{H}_0 &:= \mathbf{H}^{(0,p)}_{L}(\Sigma^x_T)\\
\mathbf{H}_1 &:= \mathbf{H}^{(1)}_{L}(\Sigma^y_T)
\end{align*}\;}
\STATE{Approximately solve: 
\[
\hat{W} \left(R_{p,L}(n)\mathbf{H}_0\right) = \mathbf{H}_1
\]
for $\hat{W}$ applying Algorithm \ref{alg:main_SLMESolver_alg_1}\;}
\end{enumerate}
\RETURN $\hat{W},R_{p,L}(n)$
\end{algorithmic}
\end{algorithm}
\end{figure}

\section{Numerical Simulations and Applications}
\label{sec:experiments}

In this section, we will present some numerical simulations computed using the {\bf DyNet} toolset available in \cite{FVides_DyNet}, which was developed as part of this project. The toolset consists of a collection of Python 3.10.4 programs for structured sparse identification and numerical simulation of discrete-time dynamical financial systems. 

The numerical experiments documented in this section were performed with {\rm Python} 3.10.4. All the programs written for real-world data reading, synthetic data generation, and sparse model identification as part of this project are available at \cite{FVides_DyNet}.

The numerical simulations described in this section were conducted on an Ubuntu 18.04.6 LTS server system. This system operates on a virtual machine within Hyper-V, equipped with 16 vCores of an Intel(R) Xeon(R) Gold 6238 CPU, running at 2.10 GHz (2095 MHz), and with 64 GB of RAM.

\subsection{Sparse autoregressive reservoir computers for dynamical nonlinear financial system behavior identification}
In this section, we focus on conducting numerical simulations to examine the behavior of a financial system modeled by a nonlinear dynamical system. These simulations aim to explore the intricate relationships among the interest rate (IR), investment demand (ID), and price index (PI) under two distinct scenarios. The governing equations of the model are as follows:
\begin{align}
&\dot{x}_1 = x_3+(x_2-s)x_1,\nonumber\\
&\dot{x}_2=1-cx_2-x_1^2,\nonumber\\
&\dot{x}_3=-x_1-ex_3,\nonumber\\
&x_1(0)=x_0,x_2(0)=y_0,x_3(0)=z_0.
\label{eq:NLFinancialModel}
\end{align}
Here, $x_1$, $x_2$, and $x_3$ denote the interest rate, investment demand, and price index, respectively.

As observed in \cite{RC_NLFinancialSystems}, systems of the form \eqref{eq:NLFinancialModel} can exhibit, among others behavior types, chaotic and eventually approximately periodic dynamic behavior depending on the configuration of parameters and initial conditions considered for \eqref{eq:NLFinancialModel}.

\subsubsection{Chaotic behavior identification}
For $s=3,c=0.1,e=1$, let us consider the initial conditions $x_0 = 2$,$ y_0=3$,$z_0=2$. For this configuration, one can obtain synthetic time series data $\Sigma_{12000}\subset \mathbb{R}^3$ obtained by applying a fourth-order adaptive numerical integration method to \eqref{eq:NLFinancialModel} for the configuration determined by the previous choice of parameters, obtaining an orbit's samples set $\Sigma_{12000}$ whose elements are uniformly distributed with respect to the time interval [0,120].

The training orbit's data set corresponding to the first $50\%$ of the data in $\Sigma_{12000}$, together with the remaining data used for model validation, are illustrated in Figure \ref{fig:NLFSignals1}.
\begin{figure}[!ht]
\centering
\includegraphics[scale=.3]{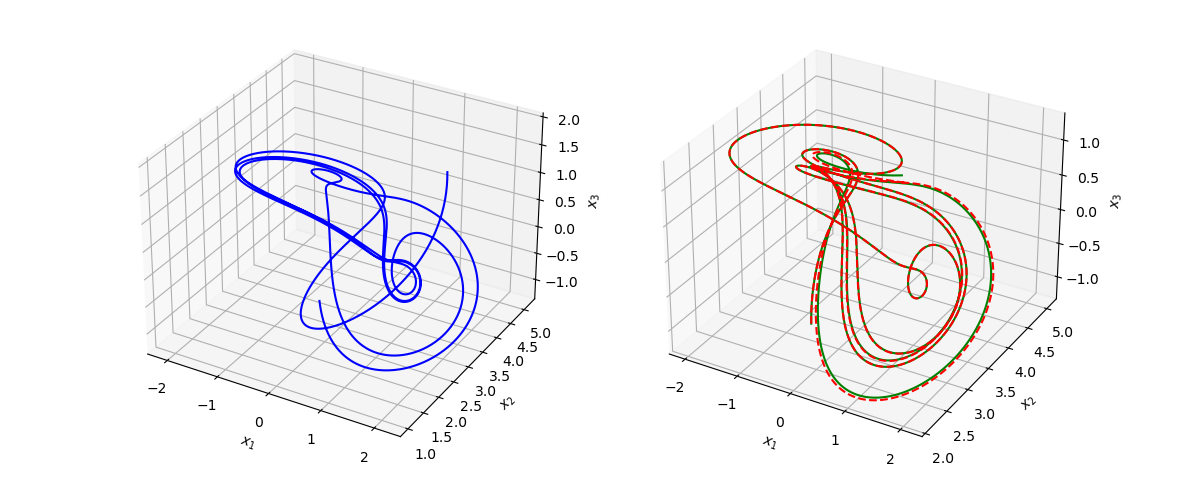}
\caption{Training orbits data (left), validation orbits data (right). The green line corresponds to validation data, and the red dotted line corresponds to the model's predictions.}
\label{fig:NLFSignals1}
\end{figure}
The factorization for the output coupling matrix $W=\hat{W}R$ determined by Theorems \ref{thm:Compression_existence} and \ref{thm:thm-1} are illustrated in Figure \ref{fig:NLFMatrices1}.
\begin{figure}[!ht]
\centering
\includegraphics[scale=.5]{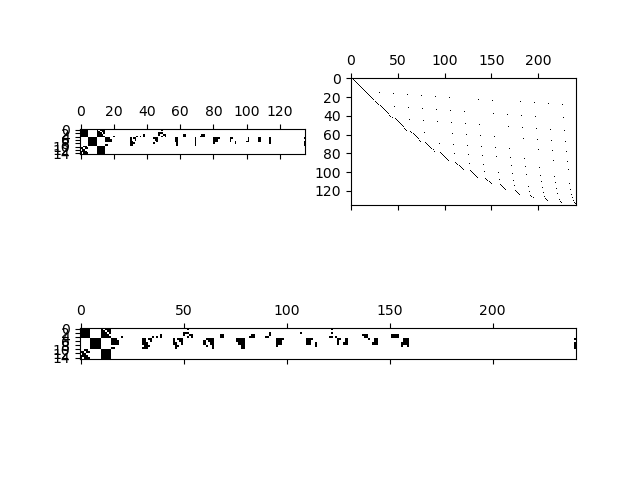}
\caption{Matrix factors $\hat{W}$ (top-left) and $R$ (top-right), output coupling matrix $W =\hat{W}R$ 
 (bottom).}
\label{fig:NLFMatrices1}
\end{figure}

\subsubsection{Eventually approximately periodic behavior identification}
For $s=0.5,c=0.1,e=0.1$, let us consider the initial conditions $x_0 = 1$,$ y_0=1$,$z_0=1$. For this configuration, one can obtain synthetic time series data $\Sigma_{12000}\subset \mathbb{R}^3$ obtained by applying a fourth-order adaptive numerical integration method to \eqref{eq:NLFinancialModel} for the configuration determined by the previous choice of parameters, obtaining an orbit's samples set $\Sigma_{12000}$ whose elements are uniformly distributed with respect to the time interval [0,120].

The training orbit's data set corresponding to the first $6.67\%$ of the data in $\Sigma_{12000}$, together with the remaining data used for model validation, are illustrated in Figure \ref{fig:NLFSignals2}.
\begin{figure}[!ht]
\centering
\includegraphics[scale=.3]{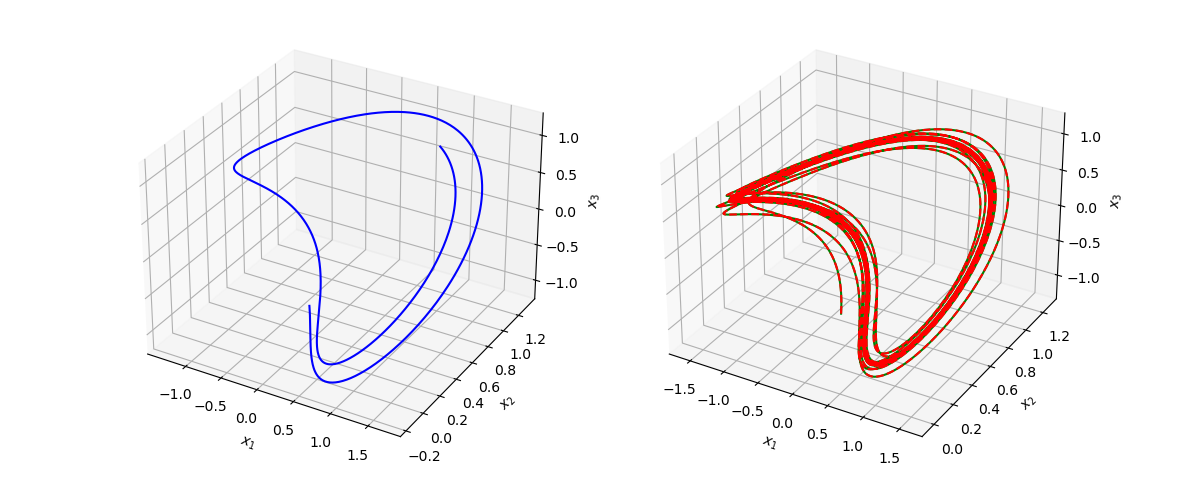}
\caption{Training orbits data (left), validation orbits data (right). The green line corresponds to validation data, and the red dotted line corresponds to the model's predictions.}
\label{fig:NLFSignals2}
\end{figure}
The factorization for the output coupling matrix $W=\hat{W}R$ determined by Theorems \ref{thm:Compression_existence} and \ref{thm:thm-1} are illustrated in Figure \ref{fig:NLFMatrices2}.
\begin{figure}[!h]
\centering
\includegraphics[scale=.5]{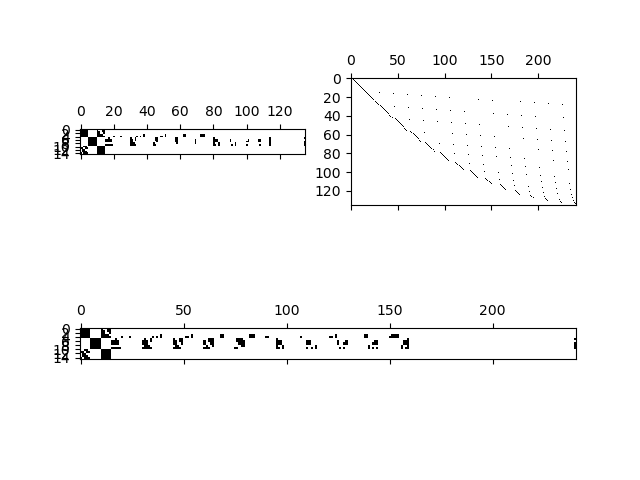}
\caption{Matrix factors $\hat{W}$ (top-left) and $R$ (top-right), output coupling matrix $W =\hat{W}R$
 (bottom).}
\label{fig:NLFMatrices2}
\end{figure}

The computational setting used for the experiments performed in this section is documented in the Python 3.10.4 program {\tt FDSExperiment.py} in \cite{FVides_DyNet} that can be used to replicate these results.

\subsubsection{Learning interest rates with sparse regressive reservoir computers.} \label{sec:linteresrate}
In this section, for financial systems described by \eqref{eq:NLFinancialModel} we will consider the problem corresponding to the identification and simulation of the interest rate signals $x_1$, when the signals $x_2,x_3$ are known.

The models considered in this section are determined by \eqref{eq:evolution_op_id} and can be described by expressions of the form:
\begin{equation}
\begin{bmatrix}
\hat{x}_1(t-1)\\
\hat{x}_1(t)
\end{bmatrix}
:=\hat{W}R_{2,2}(2)\eth_2\left(\begin{bmatrix}
        x_2(t-1)\\
        x_2(t)\\
        x_3(t-1)\\
        x_3(t)
    \end{bmatrix}\right)
     \label{eq:Reg_NLFSystem_model}
\end{equation}

When a financial system described by \eqref{eq:NLFinancialModel} exhibits an eventually periodic behavior, one can use models of the form \eqref{eq:Reg_NLFSystem_model} to learn the behavior of the interest rates, the related signals and model parameters are illustrated in Figure \ref{fig:periodic-interest-rate-id}.
\begin{figure}[!ht]
    \centering
    \includegraphics[scale=.5]{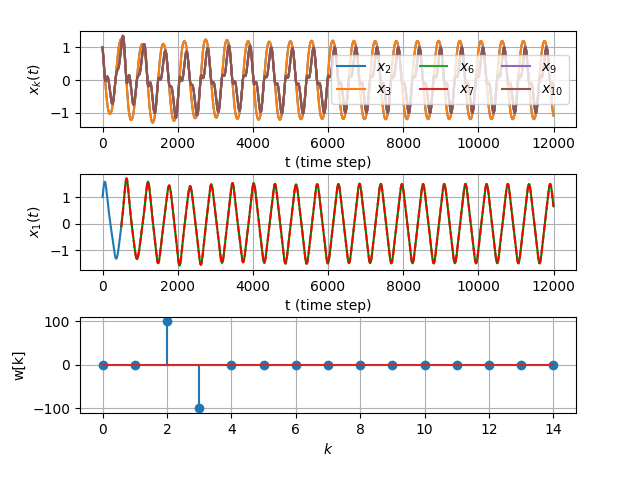}
    \caption{Periodic interest rate identification.}
    \label{fig:periodic-interest-rate-id}
\end{figure}

The Python 3.10.4 programs {\tt ARNLFinancialSystemID} and {\tt RNLFinancialSystemID} in \cite{FVides_DyNet} contain the computational settings that can be used to replicate these results.

\section{Conclusions}

The results in \S\ref{section:linear-solvers} and \S\ref{sec:structured_NAR_models} in the form of algorithms like the ones described in \S\ref{sec:algorithms}, can be effectively used for the sparse structured identification of financial dynamical models that can be used to compute data-driven predictive and prescriptive numerical simulations.

The sparse representations crucial for identifying the models' parameter matrices make SRRC models exceptionally adept at working with time series that have a relatively low volume of available training data. This attribute is particularly valuable in the financial sector, where data scarcity can often be a challenge.

From the perspective of regulatory bodies, such as the National Commission of Banks and Insurance Companies of Honduras, the inherent nature of reservoir computing in these models is a significant advantage. SRRC models are tailored to capture complex and nonlinear interactions between financial variables, offering deep insights into the interdependencies and influences within the banking system. This capability is crucial for regulatory oversight, as it aids in understanding the subtleties of market behavior and risk factors. The models provide a robust analytical tool for monitoring, regulation, and policy-making, ensuring that regulatory bodies are equipped with accurate and comprehensive analyses to oversee and guide the banking sector effectively.

\section{Future Directions}

The extension of sparse RRC modeling techniques to equivariant system identification will be studied in future communications. Further implementations of the structured sparse model identification algorithms presented in this document to compute data-driven dynamic general equilibrium models will be the subject of future communications.

\section*{Data Availability}

The programs and data that support the findings of this study will be openly available in the DyNet-CNBS repository, reference number \cite{FVides_DyNet}, in due time.\\

\section*{Conflicts of Interest}
The authors declare that they have no conflicts of interest.

\section*{Acknowledgment}

The structure preserving matrix computations needed to implement the algorithms in \S\ref{sec:algorithms}, were performed with  {\rm Python} 3.10.4, with the support and computational resources of the National Commission of Banks and Insurance Companies ({\bf CNBS}) of Honduras.

\bibliographystyle{plain}
\bibliography{SPORT_FVides.bib}

\end{document}